\documentclass[runningheads]{llncs}

\usepackage{xspace}
\usepackage[utf8]{inputenc}
\usepackage{graphicx}
\usepackage{amsmath}
\usepackage{float}
\usepackage{amsfonts}
\usepackage{algorithm}      
\usepackage{algpseudocode}  
\usepackage[breaklinks,bookmarks=false]{hyperref}
\usepackage{enumerate}
\usepackage{paralist}
\usepackage{enumitem}

\newcommand{\PINS}{PINS\xspace}
\newcommand{\PINSlong}{predecessor in nested sets\xspace}
\newcommand{\PINSlongcap}{Predecessor in nested sets\xspace}
\newcommand{\PISNS}{PISNS\xspace}
\newcommand{\PISNSlong}{predecessor in shrinking nested sets\xspace}
\newcommand{\Oh}{\mathcal{O}}
\newcommand{\wordSize}{\mathcal{W}}
\newcommand{\polylogws}{\text{\hspace{1mm}polylog}}
\newcommand{\polylog}{\text{polylog}}

\DeclareMathOperator{\link}{sl}
\DeclareMathOperator{\cost}{c}
\DeclareMathOperator{\pred}{pred}

\title{Weighted ancestors in suffix trees}

\author{Pawe{\l} Gawrychowski\inst{1} \and Moshe Lewenstein\inst{2} \and Patrick K. Nicholson\inst{1}}
\institute{Max-Planck-Institut für Informatik, Saarbrücken, Germany \and
Bar-Ilan University, Israel}

\begin{document}

\maketitle

\begin{abstract}
The classical, ubiquitous, {\em predecessor problem} is to construct a
data structure for a set of integers that supports fast predecessor
queries. Its generalization to weighted trees, a.k.a. {\em the
  weighted ancestor problem}, has been extensively explored and
successfully reduced to the predecessor problem. It is known that any
solution for both problems with an input set from a polynomially
bounded universe that preprocesses a weighted tree in
$\Oh(n\polylogws(n))$ space requires $\Omega(\log\log n)$ query
time. Perhaps the most important and frequent application of the
weighted ancestors problem is for suffix trees. It has been a
long-standing open question whether the weighted ancestors problem has
better bounds for suffix trees. We answer this question positively: we
show that a suffix tree built for a text $w[1..n]$ can be preprocessed
using $\Oh(n)$ extra space, so that queries can be answered in
$\Oh(1)$ time. Thus we improve the running times of several
applications. Our improvement is based on a number of data structure
tools and a periodicity-based insight into the combinatorial structure
of a suffix tree.
\end{abstract}

\section{Introduction}

The well-known and widely-used {\em predecessor problem} is to
preprocess a set of integers so that the predecessor of a given number
can be located. Tight tradeoffs between construction space and query
times for such a data structure are known; see P\u{a}tra\c{s}cu's
survey on predecessor search~\cite{Patrascu08}. The predecessor
problem was generalised to trees by Farach and
Muthukrishnan~\cite{FarachHashing}. It is called the \emph{weighted
  ancestor problem} and is defined as follows. We are given a rooted
tree in which every node $v$ has an associated integer \emph{weight}
$w(v)$ as input.  The weights satisfy the min-heap property, that is
the weight of every node is larger than the weight of its parent (the
tree need not be binary). The goal of the problem is to preprocess the
tree so that the predecessor of a given number, among the weights of
all the ancestor nodes of a given leaf, can be located. Farach and
Muthukrishnan~\cite{FarachHashing} give a randomised data structure,
which can be constructed in $\Oh(n)$ time and space plus the time and
space for a predecessor data structure storing $n$ integers from
$[1,U]$ given an $n$-node tree with weights from $[1,U]$. The query
time is $\Oh(\pred(n,U))$, where $\pred(n,U)$ is the predecessor query
time. Amir {\it et al.}~\cite{ALLS07} present a {\em deterministic}
version of the structure.

In the simpler unweighted version of the problem, called the
\emph{level ancestor problem}, we must preprocess a tree on $n$ nodes,
so that we can retrieve the $k$-th ancestor of a given node
efficiently. Berkman and Vishkin showed that such a query can be
answered in $\Oh(1)$ time, using $\Oh(n)$ preprocessing time and
space~\cite{Berkman}.  Later, a much simpler solution was discovered
by Bender and Farach-Colton~\cite{LevelAncestor}.  A dynamic version
has also been studied, where new leaves can be added to the
tree~\cite{Alstrup,Dietz}.  However, the solutions for level ancestor
strongly use the fact that the difference in ``weight'' between levels
is one, and therefore gives no insight into the weighted ancestors
problem.

The application for which the {\em weighted} ancestor problem was
initially introduced is {\it substring
  hashing}~\cite{FarachHashing}. In substring hashing one wants to
preprocess a given string $w[1..n]$, to allow the efficient
computation of the hash $h(i,j)$ of any of its substrings $w[i..j]$.
The hashing should be perfect, i.e., $h(i,j)=h(i',j')$ if and only if
$w[i..j]=w[i'..j']$. In~\cite{FarachHashing} the substring hashing
problem was reduced to weighted ancestor queries on a suffix
tree. Since the universe size is $\Oh(n)$ for suffix trees, one can
use a predecessor data structure such as a $y$-fast
trie~\cite{Willard83} to obtain $\Oh(n)$ preprocessing time and space
so that any hash can be computed in $\Oh(\log\log n)$ time. Their
solution also gives the same bounds for the weighted ancestor problem
in {\em any} tree where the weights are polynomial in $n$.

In the context of suffix trees the weighted ancestors problem can also
be viewed as preprocessing a suffix tree built for a string $w[1..n]$,
so as to allow the retrieval of the (implicit or explicit) node
corresponding to any substring $w[i..j]$, given $i$ and $j$. There are
numerous applications and we will mention a few later.

The weighted ancestor problem was generalised by Kopelowitz and
Lewenstein~\cite{KopelotDynamicWeighted}, who considered the dynamic
setting and showed how to support leaf insertions and edge splitting
operations (required to maintain a suffix tree for a growing
text). They showed that, up to an additive $\Oh(\log^\star n)$ term,
the static problem is as easy as predecessor search: if one can implement
a linear space static predecessor structure with a query time of
$\pred(n,U)$, then a weighted ancestor query can be answered in
$\Oh(\pred(n,U)+\log^{\star}n)$ time after linear preprocessing.  A
variant of the weighted ancestor problem was also considered by
Alstrup and Holm~\cite{Alstrup}, in which $U=\polylog(n)$, making the
situation significantly simpler.

Since the weighted ancestor is a generalisation of the predecessor
problem, it cannot have better time/space bounds than the predecessor
problem. Hence, by the known bounds for the weighted ancestor problem
in which the universe size is polynomially bounded in $n$, any
weighted ancestor data structure of size $\Oh(n\polylogws(n))$ must
have query time of $\Omega(\log\log n)$.  Furthermore, this lower
bound holds \emph{even when} the the node weights are bounded by $n$,
see Appendix~\ref{sec:lb}. Nevertheless, weighted ancestors on suffix
trees are a special case of the general weighted ancestors
problem. Hence, it is plausible that one can do better. This was
indirectly expressed by Farach and Muthukrishnan~\cite{FarachHashing}
where the question was raised whether batched substring hashing can be
sped up. This led to the challenge of solving weighted ancestors on
suffix trees in $\Oh(n)$ preprocessing time and space and $\Oh(1)$
query time which has been an open question for a long time now.

{\bf Contribution} All our results hold in the word-RAM model with
logarithmic word size.  We show that, for weighted ancestor in suffix
trees, it is possible to achieve $\Oh(1)$ deterministic worst-case
query time using $\Oh(n)$ additional space.

To sidestep the lower bound for the weighted ancestor problem, we look
deeper into the structure of a suffix tree, and apply a
periodicity-based argument. This argument allows us to decompose the
tree into sufficiently simple subtrees, which are then preprocessed
separately. To preprocess the subtrees, we develop an efficient
solution for a variant of the predecessor problem, in which we are
given multiple correlated sets of integers. The correlation allows us
to circumvent the predecessor lower bound, which would be relevant if
we were to consider each of the sets separately.  As our solution
contains many details, we provide a high level overview in
Section~\ref{sec:overview}.  This yields improved query times to
several problems.

\begin{description}[leftmargin=0pt]

\item {\bf Substring Search} Preprocess the suffix tree built for
  $w[1..n]$ to answer substring search queries, i.e., given a pair of
  indices $i,j$ return the locus of $w[i..j]$ in the suffix tree (the
  node at the end of the partial path denoting $w[i..j]$). This is
  solved by a weighted ancestor query on a suffix tree: go to the node
  representing $w[i..n]$ and answer the predecessor query of $j-i+1$
  (in this case we prefer the analogous successor query). Since
  weighted ancestors take $\Oh(n)$ preprocessing space and $\Oh(1)$
  query time, substring search has the same bounds.

\item {\bf Substring Hashing} We define $h(i,j)=\langle$locus of
  $w[i..j], i-j+1\rangle$, where the locus of $w[i..j]$ is found by
  substring search. It is easy to verify that $h(i,j)=h(i',j')$ iff
  $w[i..j]=w[i'..j']$. Hence, substring hashing can be improved to
  $\Oh(n)$ preprocessing space and $\Oh(1)$ query time. Consequently,
  batched substring hashing is also optimal. By not insisting that we
  return the corresponding node of the suffix tree, one can achieve an
  optimal $\Oh(1)$ query after $\Oh(n)$ preprocessing with a simpler
  method by Gawrychowski~\cite{Gawrychowski}.  Nevertheless, the
  number of bits in the answer by Gawrychowski~\cite{Gawrychowski} is
  $3\log n$ while in our solution it is $2\log n$, which is optimal.
  Moveover, in some applications we want to access the suffix tree
  node, as it provides more information.  For instance, we can then
  report all of its occurrences, or the leftmost occurrence.

\item{\bf Cross-Document Pattern Matching} Index a collection of
  documents, so that given a substring $w[i..j]$ of the $k$-th
  document, we can search for its appearances in the $k'$-th
  document. This problem was introduced by Kopelowitz {\it et
    al.}~\cite{KKNS14}, who also considered some extensions. Their
  linear space solution uses a generalised suffix tree with weighted
  ancestor queries. With our result the query time becomes
  $\Oh(1)$. The improvement can be also embedded in the extensions.

\item {\bf Searching Substrings Internally in the Suffix Tree} Cole
  {\it et al.}~\cite{CGL04}, when proposing data structures for
  indexing a text $w[1..n]$ with $k$ mismatches, $k$ errors and $k$
  wildcards, suggested the LCP data structure. The LCP data structure
  comes in two variants, {\em rooted LCP} and {\em unrooted LCP}. The
  former preprocesses an arbitrary collection of suffixes of $w[1..n]$
  in $\Oh(n)$ space and allows a search from the root of the
  compressed trie of these suffixes in $\Oh(\log\log n)$ time by using
  weighted ancestor queries on a careful decomposition of the
  compressed trie. The latter preprocesses such collection in
  $\Oh(n\log n)$ space to allow a search from an arbitrary node with
  an even more detailed decomposition. Both have query time
  $\Oh(\log\log n)$ because of the weighted ancestors. Alas, reducing
  this to $\Oh(1)$ is problematic because the compressed trie is not a
  suffix tree, and the nice properties that we need are
  lost. Nevertheless, we can support $\Oh(1)$ time rooted and unrooted
  LCP queries on the original suffix tree.

\item {\bf Indexing with $k$ Wildcards} In Cole {\it et
  al.}~\cite{CGL04} there is an implicit solution to indexing with $k$
  wildcards in $\Oh(n\log n)$ space, that supports queries in time
  $\Oh(m+\sigma^k\log\log n + occ)$, using the LCP data structures
  mentioned above. The space was improved to $\Oh(n)$ by Bille {\it et
    al.}~\cite{BGVV12}. Recently, in~\cite{LNV14} the running time was
  improved to $\Oh(m+\sigma^k\sqrt{\log\log\log n} + occ)$. Now this
  can be further improved to $\Oh(m+\sigma^k + occ)$ with unrooted LCP
  queries on the suffix tree itself.  The space
  improvements~\cite{BGVV12,LNV14} do not immediately carry over.

\item {\bf Fragmented Pattern Matching} The problem of {\em Substring
  Concatenation}, defined by Amir {\it et al.}~\cite{ALLS07}, requires
  preprocessing a text $w[1..n]$ so that given $i,j$ and $i',j'$we can
  return a substring of $w$ which is the concatenation of $w[i..j]$
  and $w[i'..j']$. Amir {\it et al.}~\cite{ALLS07} solved this by using
  a suffix tree, a reversed suffix tree, weighted ancestor queries on
  both and a node intersection data structure, all in $\Oh(\log\log n)$
  time. However, this can also be solved with a couple of LCP data
  structures, one rooted and one unrooted. Combined with our new
  result this achieves $\Oh(1)$ query time. The more general
  {Fragmented Pattern Matching} requires proprocessing a text
  $w[1..n]$ so that after receiving a collection of $k$ substrings as
  pairs of indices, one can answer whether there exists a substring
  $w[i_1..j_1]w[i_2..j_2]\hdots w[i_k..j_k]$ within the text. By
  extending the result for substring concatenation this takes $\Oh(k)$
  time.

\item {\bf Weighted Ancestors in Arbitrary Trees} In the process of
  solving our problem, we remove the additive $\Oh(\log^{\star}n)$
  term from the solution of Kopelowitz and
  Lewenstein~\cite{KopelotDynamicWeighted}, improving the cost of
  weighted ancestor queries in \emph{any tree} to $\Oh(\pred(n,U))$
  after linear preprocessing.  This improvement may be important in
  other cases where the node weights are not arbitrary and predecessor
  lower bounds do not apply.

\end{description}

\section{Preliminaries}

A \emph{suffix tree} of a string $w$, denoted $ST(w)$, is a compacted
trie containing all suffixes of $w\$$, where $\$$ is a unique
character not occurring in $w$.  A \emph{generalised suffix tree} of a
collection of strings $w_1,w_2,\ldots,w_k$, denoted
$GST(w_1,w_2,\ldots,w_k)$, is a compacted trie containing all suffixes
of $w_1\$_1$, $w_2\$_2$, $\ldots$, $w_k\$_k$, where each $\$_i$ is a
unique character not occurring in any of the strings. We will often
use $w_{i}[j..]$ to denote the suffix of $w_{i}$ starting at the
$j$-th character.  In a compacted trie we define the \emph{depth} of a
node to be its number of explicit ancestors, and the \emph{string
  depth} to be the length of the string it represents. In a
(generalised) suffix tree we define the \emph{suffix link} of a node
representing the string $as$ to be a pointer to the node representing
$s$. Every explicit node $v$ stores such a link $\link(v)$. If $v$
is implicit, then $\link(v)$ is not stored, but we will use this
notion in some proofs.

We want to preprocess a suffix tree built for a string $w[1..n]$, so
that, later, we can quickly retrieve the node corresponding to any
substring $w[i..j]$. If the node is explicit, then we simply return a
pointer to it, and if it is implicit, then we return a pointer to the
corresponding edge of the suffix tree.  We call this special case of
the weighted level ancestor problem {\it substring retrieval}.

We say that a natural number $p$ is a period of string\footnote{We use
  the term \emph{string} rather than \emph{word}---as is common in the
  combinatorial setting---to avoid confusion with machine words in our
  RAM.}  $w$ if $w[i]=w[i+p]$ for every $i$ such that both sides are
defined. The smallest such $p$ is called the period of $w$, and if the
period is at most $|w|/2$ we call $w$ periodic. Otherwise it is
aperiodic. The well-known property of periods is that if $p$ and $q$
are both periods of $w$, and additionally $p+q\leq |w|$, then
$\gcd(p,q)$ is a period of $w$, too. A \emph{cyclic rotation} of a
string $w[1..n]$ is a string $w[i+1..n]w[1..i]$. A \emph{Lyndon word}
has the property that it is lexicographically smallest among all its
cyclic rotations. A string $u$ is \emph{primitive} if it cannot be
represented as $u=v^{i}$ with $i>1$. The Lyndon rotation of a
primitive string is its unique Lyndon word.

All space bounds are measured in machine words, and all time bounds
are deterministic worst-case.  The following result is known to hold
in the word-RAM model with logarithmic word size: after linear
preprocessing, a collection of dynamic sets, each containing at most
$\polylog(n)$ integers at any moment, can be maintained in a linear
space structure allowing insertions, deletions and predecessor
searching in any of the sets in $\Oh(1)$ time~\cite{Fredman}.  In our
setting, all the sets will be static hence a much earlier (and
simpler) version of this result suffices~\cite{Ajtai}. Furthermore,
and the integers are from $[1,n]$, hence a simpler implementation
suffices~\cite{Grossi}. We will call the resulting structure an {\it
  atomic heap} even though that name technically refers to the more
powerful structure.

We also make use of the following result, known to hold in the
word-RAM model with logarithmic word size. Given a bit vector of
length $\Oh( \log n)$, we can perform the following
operations~\cite{J89} in $\Oh(1)$ time: $\texttt{rank}(i)$, which
returns the number of $1$ bits up to position $i$; and
$\texttt{select}(i)$, which returns the index of the $i$-th $1$ bit.
This can be done via table lookup by storing a \emph{universal table}
of size $\Oh(n^\varepsilon)$ space, for any $\varepsilon > 0$.  We
assume that we have access to such a table.

Finally, we use the solution to the level ancestor
problem~\cite{LevelAncestor}, so that we can retrieve a node of our
compacted trie as soon as we know its depth.  It is known that any
tree can be preprocessed in $\Oh(n)$ space, so that any such query can
be processed in $\Oh(1)$ time~\cite{LevelAncestor}. By applying this
result to our compacted tries, we can retrieve a node as soon as we
know its depth.  Hence we need only focus on computing the depth.

\section{\label{sec:overview}Intuition and overview}

We start by presenting the intuition behind our solution and an
overview of its formalisation, which is the main contribution of the
paper. To make the presentation easier to understand, we first
describe a simpler solution to the substring retrieval problem that
occupies $\Oh(n \polylogws(n))$ space, and allows $\Oh(1)$ time
queries.  Later, in Section~\ref{sec:space-reduction}, we present the
details to reduce the space usage to $\Oh(n)$.

Our goal is to preprocess the set of ancestors of every leaf. More
precisely, if $D(v)$ is the set of string depths of all of the
ancestors of $v$, we want to perform a predecessor search in any
$D(v)$, where $v$ is a leaf. We could preprocess every such $D(v)$
separately, but then the best query time that we can hope for is
$\Oh(\log\log n)$, assuming that the allowed preprocessing space for
every $D(v)$ is $\Oh(|D(v)| \polylogws(|D(v)|))$~\cite{Patrascu08}. To
overcome this, we observe that the sets corresponding to different
leaves $v$ are correlated.  More precisely, if we consider two leaves
$v$ and $u=\link(v)$, then $x\in D(v)$ and $x>0$ implies $x-1\in
D(u)$.  If for every leaf corresponding to a suffix $w[i..n]$ we
define a set $S_i=\{i+x : x\in D(v) \}$, then we get a collection of
sets $S_1,S_2,\ldots,S_n\subseteq[1,N]$ such that $S_i\cap
[n_{i+1},N]\subseteq S_{i+1}$, where $N=n+1$ and $n_i=i$, see
Figure~\ref{fig:example}.  We call the problem of supporting
predecessor queries on these sets {\it \PISNSlong}. In
Section~\ref{sec:predecessor} we show that such collections can be
processed using $\Oh(N\log^2 N+\sum_{i}|S_{i}|)$ space so that
predecessor searching in any $S_{i}$ takes just $\Oh(1)$ time (the
space is further improved in Section~\ref{sec:space-reduction}). So,
the correlation between different sets allows us to circumvent the
known lower bound for near-linear space predecessor structures.  Now
if it were the case that $\Oh(\sum_{i}|S_{i}|)\in
\Oh(n\polylogws(n))$, we would be done.

\begin{figure}[t]
\centering
\includegraphics[width=\textwidth]{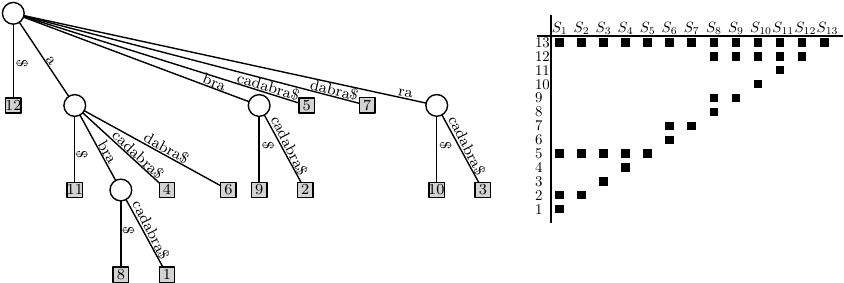}
\caption{\label{fig:example}A suffix tree built for
  \texttt{abracadabra} and the corresponding
  instance of \PISNS. }
\end{figure}

Unfortunately, it can happen that a single explicit node contributes
to multiple $D(v)$'s, hence the sum might be substantially
larger. However, when we try to construct a string with such large
sum, it seems that the most natural candidates are very repetitive,
for example $\texttt{a}^{n-1}\texttt{b}$. This is not a coincidence.
If the same explicit node $u$ contributes to two different sets $D(v)$
and $D(v')$, and the string depth of $u$ is at least $\frac{3}{4}n$,
then there are two different suffixes of the whole string such that
their longest common prefix is of length at least $\frac{3}{4}n$.
This means that the period of the \emph{middle part} of the string,
i.e., $w[\frac{1}{4}n..\frac{3}{4}n]$, is at most $\frac{1}{4}n$, or
in other words the middle part is periodic. So the intuition is that
the larger $\Oh(\sum_{i}|S_{i}|)$ is, the more periodic the
string---or at least its large part---is. 

If the period of the whole string is $p$, then (by the periodicity
lemma) any two suffixes $w[i..]$ and $w[j..]$ either branch out at
string depth less than $2p$, or the shorter suffix is a prefix of the
longer one \emph{and} $j=i+\alpha p$.  Moreover, any explicit node at
string depth less than $2p$ is the lowest common ancestor of two
leaves corresponding to suffixes of length less than $3p$.  Hence the
whole suffix tree can be decomposed into the top part, which is the
suffix tree built for the length $3p$ suffix of $w$, and $p$ long
paths corresponding to the longer suffixes starting at different
offsets modulo $p$ (with one leaf attached to every explicit node on
such path).  On the $i$-th path, all explicit nodes are at string
depths $i+\alpha p$, so it is trivial to answer a predecessor query in
$\Oh(1)$ time there. If we additionally preprocess the top part, which
is easy if $p$ is small, we can answer any predecessor query.  Hence
the intuition is that the larger the sum becomes, the less interesting
the tree is. Unfortunately, formalising this intuition is quite
technical, as we need more control on how we measure the
repetitiveness of our string: looking at its periodicity alone is not
sufficient.

To make the formalisation easier, in Section~\ref{sec:reduction} we
reduce the substring retrieval problem to a more structured variant.
In {\it long substring retrieval} we must preprocess a generalised
suffix tree $T$ built for a collection of strings
$w_{1},w_{2},\ldots,w_{\beta}$, where $\beta = \Oh(n /\ell)$, all of
the same length $\ell$, so that we can retrieve the node corresponding
to any substring $w_{k}[i..j]$ of length at least
$\frac{3}{4}\ell$. We call each $w_{i}$ a \emph{document}.  All
documents will be substrings of $w$, hence we specify them by giving
their starting and ending positions.  We show that if we can
preprocess such a collection in $S(n)$ space achieving $\Oh(1)$ query
time for the \emph{long} substring retrieval problem, then we can
solve the original substring retrieval in $\Oh((S(n)+n)\log n)$ space
and the same query time. The idea is to decompose the string into
fragments of length roughly $2^i$ for $i=0,1,\ldots,\log n$.

In Section~\ref{sec:long} we solve long substring retrieval.  We
partition the substrings of length at least $\frac{3}{4}\ell$ of all
documents into two types depending on whether their period is at least
or at most $\frac{1}{4}\ell$. Informally, both types are easy to deal
with, but for different reasons. Observe that if a substring of some
$w_{i}$ has period at most $\frac{1}{4}\ell$, then the middle part of
$w_{i}$ of length $\frac{1}{2}\ell$ is periodic. This allows us to
quickly detect if the period of $w_{i}[j..k]$ that we query with is at
least $\frac{1}{4}\ell$.

The simple case is when no $w_{i}$ has a periodic middle part, i.e.,
all substrings of length at least $\frac{3}{4}\ell$ have periods at
most $\frac{1}{4}\ell$. This implies that no $w_i$ has two suffixes of
length at least $\frac{3}{4}\ell$ such that their longest common
prefix is of length at least $\frac{3}{4}\ell$. We define $T'$ to be
the \emph{bottom part} of $T$ consisting of all nodes at string depth
at least $\frac{3}{4}\ell$. The number of leaves in any subtree of
$T'$ is exactly the number of \emph{different} documents with suffixes
in that subtree.  Additionally, we partition the nodes of $T'$ into
\emph{levels} according to the rounded logarithm of the number of
documents in their subtree.  Since this number is equal to the number
of document ending in the subtree, the nodes at the same level
constitute a collection of disjoint paths. Also, by looking at the
suffix links we observe that the explicit nodes on these paths are, in
a certain sense, nested. We exploit this nesting to retrieve the node
lying on any of these paths in constant time. This is done by reducing
the problem to \PISNSlong, allowing us to sidestep predecessor lower
bounds.

In the general case some $w_{i}$ might have a periodic middle
part. Then $T'$ is also the bottom part of $T$, but we additionally
prune it to contain only the nodes such that their subtree does not
contain the same document twice. We preprocess the pruned tree $T'$ as
in the simple case, which allows us to retrieve the node if the period
of $w_{i}[j..k]$ is larger than $\frac{1}{4}\ell$. To process
$w_{i}[j..k]$ with period at most $\frac{1}{4}\ell$, we group all
substrings of length at least $\frac{3}{4}\ell$ with period at most
$\frac{1}{4}\ell$ according to their periods.  More precisely, for
such $w_{i}[j..k]$ with period $p$ we find the (unique) Lyndon word
$r$ such that $|r|=p$ and $w_{i}[j..k]$ is a substring of
$r^{\infty}$. For every possible $r$ we build a separate structure
allowing us to locate the node corresponding to any $w_{i}[j..k]$ of
length at least $\frac{3}{4}\ell$ being a substring of $r^{\infty}$.
The structure is again based on the observation that the explicit
nodes in the corresponding part of $T$ are in a certain sense nested.

\section{\label{sec:predecessor}\PINSlongcap}

In this section we develop an efficient solution for a certain variant
of the predecessor problem, where we want to preprocess a collection
of sets of integers as to allow predecessor searching in any of
them. By predecessor searching we mean returning the rank of the
element which is the predecessor of a given value.  We start with a
version where the sets are $S_{1},S_{2},\ldots,S_{k}\subseteq [1,N]$
and $S_{i}\subseteq S_{i+1}$, which we call {\it \PINSlong} or
\emph{\PINS}.

\begin{lemma}
\label{lemma:nested}
\PINS can be solved in $\Oh(N\log N+\sum_{i}|S_{i}|)$ space and
$\Oh(1)$ time.
\end{lemma}

\begin{proof}
We partition the collection of sets into $\log N$ groups. The $k$-th
group contains all $S_{i}$ with $|S_{i}|\in[2^{k},2^{k+1})$. Because
  $|S_{i}|\leq |S_{i+1}|$, we have that the $k$-th group contains the
  sets $S_{g_{k-1}+1},\ldots,S_{g_{k}-1},S_{g_{k}}$, where
  $0=g_{0}\leq g_{1}\leq\ldots\leq g_{\log N}$. For every such group
  we allocate a table of length $N$, where we explicitly store the
  predecessor of every $x\in [1,N]$ in $S_{g_{k}}$. These tables allow
  us to locate the predecessor of any $x\in[1,N]$ in the last set of
  any group in $\Oh(1)$ time. Additionally, for every set $S_{i}$
  belonging to the $k$-th group we allocate a table of length
  $|S_{g_{k}}|$, where we store the predecessor of every $x\in
  S_{g_{k}}$ in $S_{i}$. To locate the predecessor of $x\in [1,N]$ in
  $S_{i}$, we first locate its predecessor $y$ in $S_{g_{k}}$.  Then
  we locate the predecessor of $y$ in $S_{i}$. Both steps take
  $\Oh(1)$ time using the precomputed tables. Furthermore, the table
  allocated for every $S_{i}$ is of length $|S_{g_{k}}|\leq 2|S_{i}|$,
  making the total space usage $\Oh(N\log N+\sum_{i}|S_{i}|)$.  \qed
\end{proof}

\begin{figure}[t]
\centering
\includegraphics[scale=0.73]{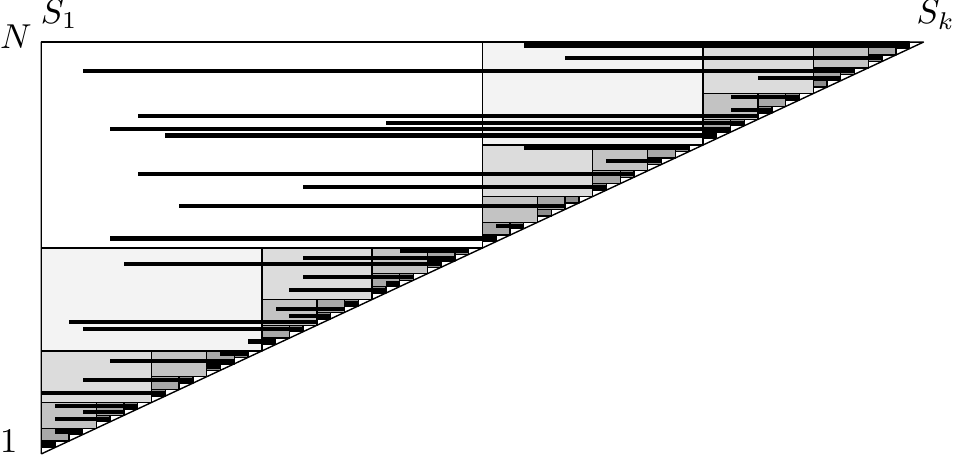}
\caption{\label{fig:triangle}An instance of \PISNS.  Each rectangle
  is an instance of \PINS.  Each horizontal line represents an
  element. If $x\in S_{i}$ then $x\in S_{j}$ as long as $j>i$ and $n_{j}\le x$.}
\vspace{-0.3cm}
\end{figure}

Now we discuss the more involved version of the problem, where we
relax the requirement that $S_{i}\subseteq S_{i+1}$. In {\it
  \PISNSlong} (or \emph{PISNS}) the sets have the additional property
that one can choose $N=n_{1}\geq n_{2}\geq \ldots \geq n_{k}$ such
that $S_{i}\subseteq [n_{i},N]$, $S_{i} \cap [n_{i+1},N] \subseteq
S_{i+1}$, and each $S_i \neq \emptyset$. We reduce this problem to a
number of carefully chosen instances of \PINS, as illustrated in
Figure~\ref{fig:triangle}.

\begin{lemma}\label{lemma:shrinking-nested}
\PISNS can be solved in $\Oh(N\log^{2}N+\sum_{i}|S_{i}|)$
space and $\Oh(1)$ time.
\end{lemma}

\begin{proof}
We decompose the problem into a number of instances of \PINS in a
recursive manner. We choose $k'$ such that $n_{k'}\leq\frac{N}{2}$ and
either $\frac{N}{2}< n_{k'+1}$ or $k'=k$. Then for every
$i=1,2,\ldots,k'$ we define $S'_{i}=\{x-\frac{N}{2}: x\in S_{i}\cap
(\frac{N}{2},N]\}$. It is easy to see that
$S'_{1},S'_{2},\ldots,S'_{k'}\subseteq [1,\frac{N}{2}]$ is an instance
of \PINS, which by Lemma~\ref{lemma:nested} can be preprocessed using
$\Oh(N\log N+\sum_{i}|S'_{i}|)$ space. Then we recursively repeat the
construction on $S_{1},S_{2},\ldots,S_{k'}\subseteq
[1,\frac{N}{2}]$. Similarly, we recursively repeat the construction on
$S_{k'+1}\cap [\frac{N}{2}+1,N], S_{k'+2}\cap [\frac{N}{2}+1,N],
\ldots, S_{k}\cap [\frac{N}{2}+1,N]$, but here we additionally
subtract $\frac{N}{2}$ from the elements as to ensure that the sets we
recurse on are from $[1,\frac{N}{2}]$.

The total size of all sets we repeat the construction on is
$\sum_{i\leq
  k'}|S_{i}|-|S'_{i}|+\sum_{k'<i}|S_{i}|=\sum_{i}|S_{i}|-\sum_{i}|S'_{i}|$
and the sum of the sizes of their universes is $N$. The recursion
depth is $\log N$ and the sizes of the universes at every \emph{level}
of the recursion sum up to $N$, hence the total space taken by all
instances of \PINS is $\Oh(N\log^{2}N+\sum_{i}|S_{i}|)$. Note that
each \PINS subproblem stores $\Oh(1)$ extra information, indicating
its offsets within the \PISNS instance, so that a query can be easily
remapped to the subproblem.

To locate the predecessor of $x\in [1,N]$ in $S_{i}$, we first must
identify the relevant subproblem instance of \PINS. This can be done
by storing, for every $i$, a single \emph{guide bit vector} of length
$\log N$, where the $j$-th bit is set iff the instance at the $j$-th
level contains at least one element originating from $S_{i}$. We
additionally store, for each $1$ bit in the guide bit vector, an
explicit pointer to the \PINS subproblem at that level.  Once the
level of the subproblem is known, a \texttt{select} query can be used
to find the correct pointer to follow.

Given $x$, we can use a $\texttt{select}$ query on its bits to
determine the level which contains $x$.  If the bit in the guide
bit vector corresponding to this level is $1$, then we search the \PINS
subproblem at this level.  If we find a predecessor, we are done.  In
the alternative case, suppose that the \PINS instance at the level of
$x$ does not contain the answer (i.e., there is no predecessor at that
level), or that the bit corresponding to this instance is a $0$ in the
guide bit vector.  In this case we can find the level containing the
predecessor using a single \texttt{rank} and \texttt{select} query on
the guide bit vector, and then query the \PINS instance.  Thus, the
query takes $\Oh(1)$ time overall.

The guide bit vectors occupy no more than $\Oh(\sum_i |S_i|)$ space,
since each bit vector occupies $\Oh(\log N)$ bits, and each set $S_i$
is non-empty. The additional pointers take at most $\Oh(\sum_i|S_i|)$
space, since we only store pointers to non-empty subproblems. 
\qed
\end{proof}

\section{\label{sec:reduction}Reduction to long substring retrieval}

In this section we reduce substring retrieval to long substring
retrieval, at the cost of a logarithmic factor increase in the space
bound.

\begin{lemma}
\label{lemma:reduction}
Suppose that any instance of long substring retrieval can be
preprocessed using $S(n)$ space so that a query can be answered in
$\Oh(1)$ time. Then, the general substring retrieval can be
preprocessed using $\Oh((S(n)+n)\log n)$ space so that a query can be
answered in $\Oh(1)$ time.
\end{lemma}

\begin{proof}
To preprocess $w$ for the general substring retrieval we construct a
constant number of instances of long substring retrieval for each
$k=0,1,\ldots,\log n$.  For every such $k$, the instances roughly
correspond to a decomposition of $w$ into documents of length around
$\ell=2^{k}$. For every $k=0,1,\ldots,\log n$ we first split $w$ into
disjoint substrings of length $2^{k}$, i.e., $b_{1}=w[1..2^{k}],
b_{2}=w[2^{k}+1..2^{k+1}], \ldots$, padding the last substring if
necessary. Then for every $\alpha=8,9,\ldots,15$ we create an instance
of long substring retrieval with $\ell=\alpha 2^{k}$ by taking the
documents to be of the form $w'_{i}=b_{i}b_{i+1}..b_{i+\alpha-1}$ for
$i=1,2,\ldots$, i.e., every possible contiguous sequence of $\alpha$
full blocks. Note that these documents are \emph{not} disjoint
substrings of $w$.  There are $\Oh(n/\ell)$ such documents.

Now consider a query concerning a substring $s$. We want to select $k$
and $\alpha\in\{8,9,\ldots,15\}$ such that $(\alpha-2) 2^{k}\leq |s| <
(\alpha-1)2^{k}$ and access the corresponding instance. This is always
possible, as we can compute $k$ such that $2^{k+3}\leq |s|<2^{k+4}$,
then $|s|-2^{k+3}<2^{k+3}$, so we can choose $\alpha'<8$ such that
$\alpha' 2^{k} \leq |s|-2^{k+3} < (\alpha'+1)2^{k}$, and finally take
$\alpha=8+\alpha'$.  Let $b_{i}$ be the block where $s$ starts, then
$s$ is fully within $b_{i}b_{i+1}..b_{i+\alpha-1}$, so we can query
the instance with the substring of $w'_{i}$-th document equal to $s$.
For the answer to be correct, we must guarantee that
$|s|\geq\frac{3}{4}\alpha 2^{k}$, but this follows from $\alpha\geq
8$. Hence using long substring retrieval we get the node $v$
corresponding to $s$ in the the generalised suffix tree built for all
$w'_{i}$. 

For every explicit node of the generalised suffix tree we store a
pointer to the corresponding node of the suffix tree of whole string
$w$.  If the node $v$ corresponding to $s$ in the generalised suffix
tree is explicit, then following its pointer gives us the final
answer. If $v$ is implicit, then it lies on an edge between two
explicit nodes, $v'$ and $v''$, corresponding to strings of length
strictly smaller and strictly larger than $|s|$, respectively.  Both
$v'$ and $v''$ are explicit in the suffix tree. Now we observe that
any substring of $w$, of length at most $(\alpha-1)2^{k}>|s|$, is a
substring of some document $w'_{i}$. Hence, if we look at the suffix
tree, then there are no explicit nodes between $v'$ and $v$. So, the
answer that we are seeking is determined by the topmost descendant of
$v'$ in the suffix tree with $v''$ in its subtree, which can be
preprocessed and stored for every $v''$. Therefore, we can compute the
answer in $\Oh(1)$ time, and the additional preprocessing space is
$\Oh(n)$, plus that taken by the $\Oh(\log n)$ instances of the long
substring retrieval problem.
\qed
\end{proof}

\section{\label{sec:long}Solving long substring retrieval}

In this section we develop an efficient solution for long substring
retrieval.  Recall that the goal in long substring retrieval is to
preprocess a generalised suffix tree built for documents
$w_{1},w_{2},\ldots,w_{\beta}$, where $\beta=\Oh(n/\ell)$ and
$|w_{i}|=\ell$, as to retrieve the node corresponding to any
$w_{k}[i..j]$ of length at least $\frac{3}{4}\ell$.

\subsection{Handling active nodes}

Let $T$ be the generalised suffix tree built for
$w_{1},w_{2},\ldots,w_\beta$, where $\beta=\Oh(n/\ell)$.  While the goal is to
preprocess the whole bottom part of $T$, i.e., all nodes at string
depth at least $\frac{3}{4}\ell$, we will first show how to preprocess
just some of these nodes. A node of $T$ is {\it active} if its string
depth is at least $\frac{3}{4}\ell$ and additionally there are no two
different leaves corresponding to the suffixes of the same document in
its subtree. Notice that if $v$ is not active, neither is its parent, hence
we can find a collection of nodes $v_{1},v_{2},\ldots,v_{s}$ such that
a node is active iff it is a (not necessarily proper) descendant of some $v_{i}$. The active
part of $T$, i.e., the forest consisting of all subtrees rooted at
$v_{1},v_{2},\ldots,v_{s}$, will be called $T'$.
We have the following property of active nodes.

\begin{lemma}
\label{lemma:suffix active}
If a non-root node $v$ is not active, then $\link(v)$ is not active either.
\end{lemma}

\begin{proof}
Let $v$ be a non-active node. There are two possible reasons for $v$
not being active. The first is that its string depth is smaller than
$\frac{3}{4}\ell$, in which case the string depth of $\link(v)$ is
also smaller and hence $\link(v)$ is not active either.  The second is
that the subtree rooted at $v$ contains two different leaves $u_1$ and
$u_2$ corresponding to the suffixes of the same document. However, in
this case $\link(u_1)$ and $\link(u_2)$ are two different leaves
corresponding to the suffixes of the same document and inside the
subtree rooted at $\link(v)$, hence $\link(v)$ is not active.  
\qed
\end{proof}

We will preprocess $T$ so that we can retrieve the node corresponding
to a substring $w_{k}[i..j]$ assuming that it is active. First we
observe that it is not difficult to detect that the corresponding node
is not active: for every leaf of $T$ we can compute and store the
string depth of its active ancestor that has the smallest string
depth. Then we can take the leaf corresponding to $w_{k}[i..]$ and check
if it has an active ancestor with a sufficiently large string depth.

We partition $T'$ into disjoint paths using a variant of the centroid
path decomposition.  First define the \emph{level} of a node $v\in T$
to be the unique integer $k$ such that the number of leaves in the
subtree of $v$ belongs to $[2^{k},2^{k+1})$.  From the definition, the
  level of any ancestor of $v$ is at least as large as the level of
  $v$, and any node has at most one child at the same level.  We also
  need the following properties of the levels, which are specific to
  the tree $T$. While we can afford to store the
  level only at the explicit nodes, all properties hold also for
  implicit nodes, and clearly the level of an implicit node can be
  determined by looking at its first explicit descendant.  Based on
  these definitions, we prove the following two lemmas.

\begin{lemma}
\label{lemma:suffix level}
The level of $\link(v)$ is at least as large as the level of $v$.
\end{lemma}

\begin{proof}
If the level of $v$ is $k$, then the subtree rooted at $v$ contains at
least $2^k$ different leaves $u_1,u_2,\ldots$. The nodes
$\link(u_1),\link(u_2),\ldots$ are also leaves and belong to the
subtree rooted at $\link(v)$, hence the level of $\link(v)$ is at
least $k$.
\qed
\end{proof}

\begin{lemma}
\label{lemma:one incoming}
Suppose $u$ and $v$ are two nodes at the same level, such that $u$ is
neither an ancestor or descendant of $v$.  If $\link(u)$ is an
ancestor of $\link(v)$, then its level is larger than the levels
of $u$ and $v$.
\end{lemma}

\begin{proof}
Let the level of $u$ and $v$ be $k$. Then the subtree of $u$ contains
at least $2^k$ different leaves $u_1,u_2,\ldots$, and, similarly, the
subtree of $v$ contains at least $2^k$ different leaves
$v_1,v_2,\ldots$.  Then all $\link(u_1),\link(u_2),\ldots$ and
$\link(v_1),\link(v_2),\ldots$ are leaves belonging to the subtree
rooted at $\link(u)$. Because all $u_i$ are different, so are all
$\link(u_i)$. Similarly, because $v_i$ are different, so are
$\link(v_i)$. Now we claim that it cannot happen that
$\link(u_i)=\link(v_j)$.  If it were the case, from the assumption
about the unique separators terminating every document we would have
that $u_i$ and $v_j$ correspond to two different suffixes of the same
document. But because $u$ is neither an ancestor or descendant of $v$
it must be that $u_{i}\neq v_{j}$, so then the string depths of
$u_{i}$ and $v_{j}$ are different, and so are the string depths of
$\link(u_i)$ and $\link(v_j)$.  Hence all $u_i$ and $v_j$ are
different, so $\link(u)$ contains at least $2^{k+1}$ different leaves in its
subtree, hence it level is larger than $k$.  
\qed
\end{proof}

From now on we focus on a fixed level $k$.  Since no node has two
children at the same level, the active nodes at level $k$ create a set
of disjoint paths, $p_{1},p_{2},\ldots,p_{s}$, such that no node in
$p_{i}$ is an ancestor of a node in $p_{j}$ if $i\neq j$.  Every such
path starts at an explicit node which has no child at level $k$ and
continues up, terminating either just before another explicit node at
a level larger than $k$ or an implicit node at string depth exactly
$\frac{3}{4}\ell$. We say that path $p_{i}$ \emph{points to path}
$p_{j}$, denoted $p_{i}\rightarrow p_{j}$, if there is a node $u\in
p_{i}$ and a node $v\in p_{j}$ such that $\link(u)=v$. This is a valid
definition, and furthermore any path is pointed to by at most one
other path, as shown in the following lemma.

\begin{lemma}
\label{lemma:degree one}
Relation $\rightarrow$ has the following properties
\begin{inparaenum}[(a)]
\item if $p_{i}\rightarrow p_{j}$ then $i\neq j$,
\item if $p_{i}\rightarrow p_{j}$ and $p_{i}\rightarrow p_{j'}$ then $j=j'$,
\item if $p_{i}\rightarrow p_{j}$ and $p_{i'}\rightarrow p_{j}$ then $i=i'$.
\end{inparaenum}
\end{lemma}

\begin{proof}
$ $ 
\begin{enumerate}[label=(\alph*)]
\item Assume that $p_{i}\rightarrow p_{i}$. Then there are $u,v\in
  p_{i}$ such that $u=\link(v)$. Then clearly the string depth of $u$
  is larger than the string depth of $v$, and hence $u$ is a proper
  ancestor of $v$. The subtree rooted at $u$ contains at least one
  leaf corresponding to a suffix of some document, say
  $w_{j}[k..]$. Then the subtree rooted at $u$ contains the leaf
  corresponding to $w_{j}[k+1..]$, so the subtree rooted at $u$
  contains two leaves corresponding to different suffixes of the same
  document, so $u$ cannot be active, which is a contradiction.

\item Assume that $p_{i}\rightarrow p_{j}$ and $p_{i}\rightarrow
  p_{j'}$. Then there are nodes $u,u'\in p_{i}$, $v\in p_{j}$ and
  $v'\in p_{j'}$ such that $\link(u)=v$ and $\link(u')=v'$.  We can
  assume that $u$ is an ancestor of $u'$, and it implies that $v$ is
  an ancestor of $v'$. Now we observe that because $v$ is active and
  both $v$ and $v'$ are at level $k$, in fact all nodes on the path
  from $v'$ up to $v$ are active and at level $k$, so $j=j'$.

\item Assume that $p_{i}\rightarrow p_{j}$ and $p_{i'}\rightarrow
  p_{j}$. Then there are nodes $u\in p_{i}$, $u'\in p_{i'}$ and
  $v,v'\in p_{j}$ such that $\link(u)=v$ and $\link(u')=v'$.  We can
  assume that $v$ is an ancestor of $v'$.  Then if $i\neq i'$ we have
  that $u$ is neither an ancestor or descendant of $u'$ so we can
  apply Lemma~\ref{lemma:one incoming} to $u$, $u'$ and
  $v=\link(u)$. We get that the level of $v$ is larger than $k$, which
  is a contradiction.
\end{enumerate}
\qed
\end{proof}

Hence we can partition the whole set of paths of active nodes at level
$k$ into:

\begin{enumerate}
\item \emph{cycles of paths}, which are of the form
  $p_{i_{1}}\rightarrow p_{i_{2}}\rightarrow\ldots\rightarrow
  p_{i_{z}}\rightarrow p_{i_{1}}$, $z\geq 2$;
\item \emph{chains of paths}, which are of the form
  $p_{i_{1}}\rightarrow p_{i_{2}}\rightarrow\ldots\rightarrow
  p_{i_{z}}$, where $p_{i_{z}}$ doesn't point to any path and no path
  points to $p_{i_{1}}$.
\end{enumerate}

\begin{figure}[H]
\centering
\includegraphics[width=\textwidth]{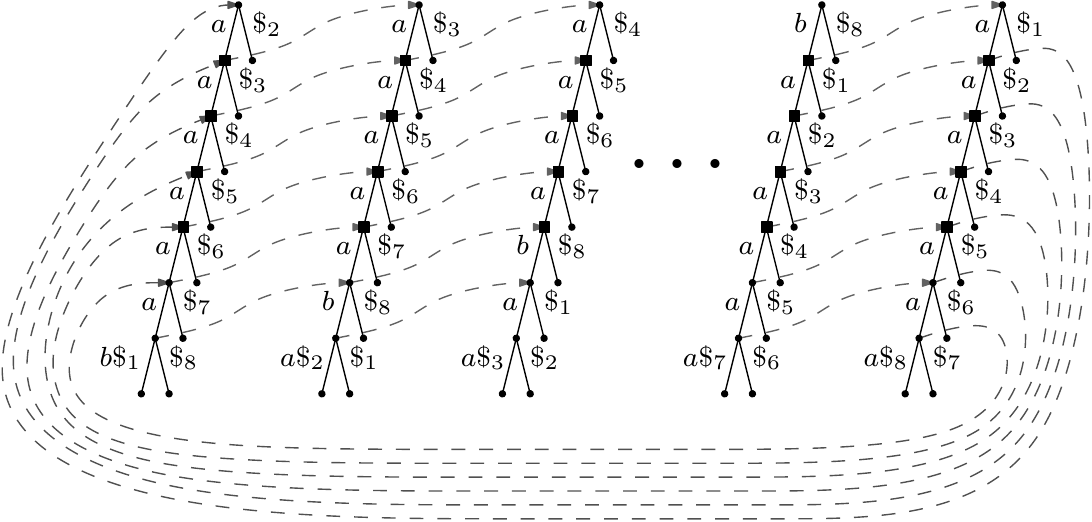}
\caption{\label{fig:cop}To visualise a cycle of paths, we can
  construct the generalised suffix tree for the set of documents $
  \{w_1, \ldots, w_\ell\}$ where $w_i = (\texttt{a}^{\ell-i}\texttt{ba}^{i-1})^4$, for
  any $\ell \ge 4$.  In the example we choose $\ell = 8$, and draw
  \emph{only the active nodes}.  Suffix links appear as dark gray
  dashed lines.  The square nodes, which are at level two, form a
  cycle of paths.}
\end{figure}

See Figure~\ref{fig:cop} for an example.  We will preprocess every
such cycle and chain separately using the solution for \PISNSlong from
the previous section, which allow us to answer a predecessor query on
any path in $\Oh(1)$ time, and bound the total space used by all the
instances of the solution.

Consider a single cycle or chain of paths, where for a cycle of paths
we additionally define $i_{z+1}=i_{1}$.  If $v\in p_{i_{j}}$ is an
explicit node, then $\link(v)$ is either an explicit node on
$p_{i_{j+1}}$, or its level is larger. Hence the sets of explicit
nodes on subsequent paths are, in a certain sense, nested. To
formalise this intuition, for every path $p_{i_{j}}$ we denote the
smallest and largest string depth of an (implicit or explicit) node by
$\ell_{j}$ and $r_{j}$, respectively. Then the range of this path is
an interval $U_{j}=[j+\ell_{j},j+r_{j}]$. Furthermore, we construct a
set $S_{j} \subseteq U_{j}$ corresponding to the path by including the
depth of every explicit node (increased by $j$ for technical
reasons). Now the ranges and the sets are nested in the following
sense.

\begin{lemma}
\label{lemma:monotonicity}
The following properties of $U_{j}$ and $S_{j}$ hold
\begin{inparaenum}[(a)]
\item $j+\ell_{j}\leq j+1+\ell_{j+1}$,
\item $j+r_{j}\leq j+1+r_{j+1}$,
\item $S_{j}\cap U_{j+1} \subseteq S_{j+1}$.
\end{inparaenum}
\end{lemma}

\begin{proof}
$ $ 
\begin{enumerate}[label=(\alph*)]
\item Assume that $\ell_{j+1} < \ell_{j}-1$. Then we have $u\in p_{i_{j}}$
  at string depth $\ell_{j}$ such that its parent $w$ is either not
  active or at a higher level, and $v\in p_{i_{j+1}}$ at string depth
  strictly smaller than $\ell_{j}-1$. We also have $u'\in p_{i_{j}}$ and
  $v'\in p_{i_{j+1}}$ such that $\link(u')=v'$. Because $u$ is the topmost
  node in $p_{i_{j}}$, $u'$ is a (not necessarily proper) descendant of
  $u$, and $v'$ is a (not necessarily proper) descendant of $v$. Hence
  $\link(w)$ is a node at string depth $\ell_{j}-2$ which is an
  ancestor of $v'$. Because $p_{i_{j+1}}$ contains the ancestors of $v'$
  up to $v$, which are at depth strictly smaller than $\ell_{j}-1$, we
  have that $\link(w)\in p_{i_{j+1}}$.  So $\link(w)$ is active and at
  level $k$. By Lemma~\ref{lemma:suffix level} the level of $w$ is at
  most $k$. Combining this with the fact that the level of its child
  $u$ is $k$, we get that the level of $w$ is exactly $k$. Hence the
  only possible reason for $w$ not belonging to $p_{i_{j}}$ is that of not
  being active. It means that either the string depth of $w$ is too
  small or the subtree rooted there contains two leaves corresponding
  to suffixes of the same document. But the string depth of $w$ is at
  least $\ell_{j+1}$, and we have some (possibly different) active
  node at such string depth, which excludes the former possibility. To
  exclude the latter, we observe that $\link(w)$ would contain two
  such leaves, so it could not belong to $p_{i_{j+1}}$.

\item Assume that $r_{j+1} < r_{j}-1$. Then we have $u\in p_{i_{j}}$ at
  string depth $r_{j}$ such that $\link(u)$ does not belong to
  $p_{i_{j+1}}$. We also have $u'\in p_{i_{j}}$ and $v'\in p_{i_{j+1}}$ such that
  $\link(u')=v'$, and $u'$ is an (proper, as otherwise $\link(u)\in
  p_{i_{j+1}}$ immediately) ancestor of $u$.  Then $\link(u')$ is an
  ancestor of $\link(u)$, so $\link(u)$ is active and at level at most
  $k$. Hence the only possible reason for $\link(u)$ not belonging to
  $p_{i_{j+1}}$ is that its level is strictly smaller than $k$, but it
  cannot happen, as the subtree rooted at $u$ contains at least
  $2^{k}$ leaves, hence so does the subtree rooted at $\link(u)$.

\item Assume that we have $x\in S_{j}\cap U_{j+1}$ but $x\notin
  S_{j+1}$. Then there is an explicit node $u\in p_{i_{j}}$ at string
  depth $x-j$ such $x\in [j+1+\ell_{j+1},j+1+r_{j+1}]$ and there is no
  explicit node at depth $x-j-1$ on $p_{i_{j+1}}$. But $\link(u)$ clearly
  is such an explicit node.  

\end{enumerate}
\qed
\end{proof}

To execute a predecessor query on $p_{i_{j}}$, it is enough to perform
such a query on the corresponding set $S_{j}$, so we focus on
preprocessing all these sets.  It is clear that their total size is
small, as every element of $S_{j}$ corresponds to a different explicit node of
$T'$, but this is not enough to beat the $\Oh(\log\log n)$
bound on the query time. We need an insight into the structure
of all $S_{j}$ based on Lemma~\ref{lemma:monotonicity}.

Suppose we extend every range to the right by defining
$U'_{j}=[j+\ell_{j},z+r_{z}]$.  Then it still holds that $S_{j}\cap
U'_{j+1}\subseteq S_{j+1}$, but additionally all $U'_{j}$ end with the
same number. We will preprocess all $U'_{j}$ using the data structure
of Lemma~\ref{lemma:shrinking-nested}. Its space usage depends on the
total size of all $S_{j}$, which as already observed is small, but
also on the size of the largest extended range $U'_{1}$.  Even though
a single $|U'_{1}|$ might be big, the sum of all such values over all
cycles and chains of paths is at most $n/2^{k}$ by the following
lemmas based on charging arguments. We define the \emph{cost}
$\cost(p_{i_{j}})$ of a path $p_{i_{j}}$ as follows:

\begin{enumerate}
\item $\cost(p_{j})=r_{j}-r_{j-1}+1$ if $j>1$, 
\item for a cycle of paths $\cost(p_{j})=r_{1}-r_{z}+1$ if $j=1$.
\item for a chain of paths $\cost(p_{j})=r_{1}-\ell_{1}+1$ if $j=1$,
\end{enumerate}

Note that for a cycle of paths we arbitrarily fix one of the paths to
be $p_{i_{1}}$.  The following two lemmas bound the costs of individual
chains (or cycles) of paths, and the cost of all paths at level $k$,
respectively.

\begin{lemma}
\label{lemma:telescope}
For any chain of paths we have that $|U'_{1}|=\sum_{j}\cost(p_{i_{j}})$,
and for any cycle of paths $|U'_{1}|\leq2\sum_{j}\cost(p_{i_{j}})$.
\end{lemma}

\begin{proof}
For a chain of path we have that
$\sum_{j}\cost(p_{i_{j}})=r_{1}-\ell_{1}+1+\sum_{j>1}( r_{j}-r_{j-1}+1)$,
which telescopes leaving only $z+r_{z}-\ell_{1}=|U'_{1}|$. Now
consider a cycle of paths. We have that $r_{1}-\ell_{1} < z$ by the
following argument.  If the inequality does not hold, then we could
take the node $u$ in $p_{i_{1}}$ at string depth $r_{1}$ and, following
the suffix links $\link(u), \link(\link(u)),\ldots$, return to $p_{i_{1}}$
after exactly $z$ steps.  This would imply that the topmost node of
$p_{i_{1}}$ contains two leaves corresponding to suffixes of the same
document. Using this inequality, we get that $|U'_{1}|\leq
z+\sum_{j>1}(r_{j}-r_{j-1}+1)$, and then because
$r_{1}-r_{z}+\sum_{j>1}(r_{j}-r_{j-1})=0$ we get
$\cost(p_{i_{1}})+\sum_{j>1}\cost(p_{i_{j}})=z$, so finally $|U'_{1}|\leq
2\sum_{j}c(p_{i_{j}})$.
\qed
\end{proof}

\begin{lemma}
\label{lemma:sum of increases}
The sum of costs of all paths at level $k$ is at most $3n/2^{k}$.
\end{lemma}

\begin{proof}
We separately bound the total cost of all paths which are first on
their respective chains, and the total cost of all the remaining ones.

Consider a path $p$ such that there is no path $p'$ for which
$p'\rightarrow p$.  Let $\ell$ and $r$ be the smallest and largest
string depth of an (implicit or explicit) node on $p$, and let $u$ be
the node corresponding to the latter.  As the level of $u$ is $k$, it
has at least $2^{k}$ different leaves $v_{1},v_{2},\ldots$ in its
subtree. Say that $v_{i}$ corresponds to $w_{a_{i}}[b_{i}..]$. Because
$u$ is active, all $a_{i}$ are different. We distribute the cost of
$p$, which is $r-\ell+1$, among the first $2^{k}$ of these suffixes by
charging $1/2^{k}$ to every letter $w_{a_{i}}[b_{i}+\ell-1],
w_{a_{i}}[b_{i}+\ell+1],\ldots,w_{a_{i}}[b_{i}+r-1]$.  Now we claim
that during this process no letter will ever be charged twice.  Assume
otherwise, so some letter $w_{i}[j]$ is charged twice to pay for two
different paths $p_{1}$ and $p_{2}$. Then there is a node $u_{1}\in
p_{1}$ corresponding to some $w_{i}[k_{1}..j]$ and a node $u_{2}\in
p_{2}$ corresponding to some $w_{i}[k_{2}..j]$.  We can assume
$k_{1}<k_{2}$. Then we can construct a sequence of nodes
$u_{1},\link(u_{1}), \link(\link(u_{1})),\ldots,u_{2}$ such that the
first and the last node are both active and at level $k$. Hence from
Lemma~\ref{lemma:suffix active} and Lemma~\ref{lemma:suffix level}
also the next-to-last node in that sequence is active and at level
$k$, and so it belongs to some path $p'$. Then $p'\rightarrow p_{2}$,
which is a contradiction. Hence no letter is charged twice, and the
total cost is $n/2^{k}$.

Now consider paths $p$ and $p'$, such that $p'\rightarrow p$.  Let
$[\ell,r]$ and $[\ell',r']$ be the ranges of string depths on nodes on
$p$ and $p'$, respectively. The last node $u$ on $p$ has at least
$2^{k}$ leaves $v_{1},v_{2},\ldots$ in its subtree, and as in the
previous case we distribute the cost of $p$, which is $r-r'+1$, among
their corresponding suffixes $w_{a_{i}}[b_{i}..]$, but now we charge
both letters and whole suffixes. 
We charge $1/2^{k}$ to every suffix $w_{a_{i}}[b_{i}..]$ and
every letter $w_{a_{i}}[b_{i}+r'],
w_{a_{i}}[b_{i}+\ell+1],\ldots,w_{a_{i}}[b_{i}+r-1]$.  Assume that
some letter $w_{i}[j]$ is charged twice for two different paths
$p_{1}$ and $p_{2}$.  As in the previous case, it implies that there
is a node $u_{1}\in p_{1}$ corresponding to some $w_{i}[k_{1}..j]$ and
a node $u_{2}\in p_{2}$ corresponding to some $w_{i}[k_{2}..j]$, and
we can construct a sequence of nodes to find a node $u'\in p''$ such
that $\link(u')=u_{2}$. But then $p''\rightarrow p$, which by
Lemma~\ref{lemma:degree one} implies that $p'=p''$, and then $r' >
|w_{i}[k_{2}-1..j]|$, so $w_{i}[j]$ is not charged by $p_{2}$. The
total number of letters and suffixes is $n$, making the total cost
$2n/2^{k}$. 
 \qed
\end{proof}

To locate the node corresponding to $w_{k}[i..j]$, we first retrieve
the leaf of $T$ corresponding to the whole $w_{k}[i..]$. Then we must
compute the level of the node corresponding to $w_{k}[i..j]$. More
precisely, we must find an ancestor $u$ of $v$ at level $k$ such that
the string depth of $u$ is at least $|w_{k}[i..j]|$, and furthermore
the level of the node corresponding to $w_{k}[i..j]$ is the same as
the level of $u$. This is enough to reduce the query to a weighted
predecessor search on a single path in one of our collections.
Computing $u$ can be done in $\Oh(1)$ using the following lemma, which
also removes the $\Oh(\log^{*}n)$ additive term from the query
complexity of~\cite{KopelotDynamicWeighted}.

\begin{lemma}
\label{lemma:marked predecessor}
A weighted tree on $n$ nodes, where some of the nodes are marked, but
any path from a leaf to the root contains at most $\Oh(\log n)$ marked
nodes, can be preprocessed in $\Oh(n)$ space so that predecessor
search can be performed among the marked ancestors of any node in
$\Oh(1)$ time.
\end{lemma}

\begin{proof}
Constructing a structure of size $\Oh(n\log n)$ is straightforward: we
store the string depths of at most $\log n$ marked ancestors of every
node in an atomic heap.  To decrease the space, we use the micro macro
tree decomposition~\cite{GabowTarjan}.  We choose $\Oh(n/\log n)$
macro nodes of the tree such that removing them leaves us with a
collection of micro trees of size at most $\log n$ each.  For every
macro node we construct an atomic heap storing the depths of all its
marked ancestors. This allows us to perform a search at any macro node
in $\Oh(1)$.  However, it might be the case that we want to perform a
search at a non-macro node. In such a case we first lookup its first
macro ancestor and do the search there. This gives us the correct
answer unless it lies within the same micro tree. Hence we need to
implement a $\Oh(1)$ time search within every micro tree.

For every micro tree we construct an atomic heap containing the string
depths of all nodes inside the micro tree. Let the sorted list of
these string depths be $d_1 \leq d_2 \leq \ldots \leq d_k$, where
$k\leq \log n$ (note that we keep all duplicates in the list). For
every node $v$ of the micro tree we store a single machine word $B[v]$
with the $i$-th bit set iff the marked ancestor of $v$ at string depth
$d_i$, if any, belongs to the same micro tree. Now to perform a search
at $v$ with a string depth $d$, we first find the predecessor of $d$
in the list. This takes $\Oh(1)$ time using the atomic heap and
assuming that every element in the atomic heap stores the position of
its first occurrence in the list. If the predecessor is $d_i$, we find
the largest $i'\leq i$ such that $B[v]$ has the $i'$-th bit set. Then
$d_{i'}$ is the predecessor of $d$ among the string depths of all
marked ancestors of $v$ inside its micro tree.  Additionally, because
the list contains duplicates, $i'$ uniquely determines the marked
ancestor corresponding to the answer. 
\qed
\end{proof}

To apply the above lemma, we mark the explicit nodes of $T$ such that
the level of their parent is strictly larger. As the maximum level is
$\log n$, the maximum number of marked nodes on any path from the leaf
is also $\log n$. Hence we have reduced the query to performing a
predecessor search among all ancestors on the same level of an
explicit node $u$. At every explicit node we store a pointer to its
path, and for every path we store a pointer to its cycle or chain of
paths.

\subsection{Handling the remaining nodes}

The method from last subsection allows us to retrieve the node $v$
corresponding to $w_{i}[j..k]$ if it belongs to $T'$, or detect that
we need to look at the non-active part.  If $v$ does not belong to
$T'$, even though $|w_{i}[j..k]|\geq\frac{3}{4}\ell$, then its subtree
contains two different leaves originating from the same document.  But
then these leaves correspond to some $w_{i'}[j'..]$ and
$w_{i'}[j''..]$ with $j'\neq j''$, and furthermore $w_{i}[j..k]$ is a
prefix of both these suffixes. It follows that the period of
$w_{i}[j..k]$ is at most $\frac{1}{4}\ell$. We preprocess all such
$w_i[j..k]$ separately.

As discussed in Section~\ref{sec:overview}, if the period of
$w_{i}[j..k]$ of length at least $\frac{3}{4}\ell$ is at most
$\frac{1}{4}\ell$, then the middle part of $w_{i}$, namely
$w_{i}[\frac{1}{4}\ell..\frac{3}{4}\ell]$, is periodic. For every
$w_{i}$ we compute the period $p$ of its middle part, and if
$p\leq\frac{1}{4}\ell$ we also find the lexicographically smallest
cyclic rotation of the corresponding string $r$ of length $p$ such
that the middle part is a substring of $r^{\infty}$.  We group
together all $w_{i}$ with the same $r$ and preprocess the subtree of
$T$ corresponding to their substrings fully contained in the periodic
part separately.

For a string $r$, let $T_r$ be the subtree of $T$ corresponding to all
substrings of $r^\infty$ of length at least $\frac{1}{2}\ell$. First
we show that any such $T_r$ can be efficiently preprocessed for
weighted level ancestor queries. In this case the input to a query is
a substring of $r^\infty$ specified by its length and starting
position.  Without loss of generality the starting position is less
than $|r|$.  

\begin{lemma}
\label{lemma:periodic preprocessing}
Let $r$ be any primitive string of length at most $\frac{1}{4}\ell$,
and $s$ be the number of explicit nodes in $T_r$ at string depth at
least $\frac{1}{2}\ell$.  $T_r$ can be preprocessed using
$\Oh(|r|\log|r|+s)$ space, so that, in $\Oh(1)$ time, the node
corresponding to any substring of $r^\infty$ of length at least
$\frac{3}{4}\ell$ can be retrieved.
\end{lemma}

\begin{proof}
For every cyclic shift $r'=r[i..|r|]r[1..i-1]$ of $r$, where
$i=1,2,\ldots,|r|$, we denote by $p_i$ the longest path in $T_r$
corresponding to a prefix of $r'^{\infty}$. Hence the whole $T_r$
can be seen as a union of these $|r|$ paths. The paths are not
necessarily disjoint, but no two of them share a common prefix
of length $2|r|$, as otherwise the periodicity lemma would imply
that $r$ is actually not primitive. We conceptually
extend every $p_{i}$ so that it corresponds to $r[i..|r|]r^{\alpha}$,
with the same value of $\alpha$ for every $i$. As we are working with
a compacted trie anyway, such an extension doesn't increase the size of
the problem.

We split every $p_{i}$ into a prefix corresponding to $r[i..|r|]r^{\alpha-\beta}$
and then $\beta$ fragments corresponding to the remaining $\beta$
repetitions of $r$. The value of $\beta$ is chosen so that the following
two conditions hold:
\begin{enumerate}
\item any explicit node that we could possibly be required to return as an
answer belongs to one of these $\beta$ fragments,
\item any explicit node belonging to one of these $\beta$ fragments
is at string depth at least $\frac{1}{2}\ell$.
\end{enumerate}
The conditions translate to $|r|-i+1+|r|(\alpha-\beta)\leq\frac{3}{4}\ell$
and $|r|-i+1+|r|(\alpha-\beta)\geq\frac{1}{2}\ell$, respectively. As
$|r|\leq\frac{1}{4}\ell$, such a $\beta$ always exists.

For every $p_i$ we define the sets $S_{i,j}$ for $j=1,2,\ldots,\beta$ describing
the string depths of all explicit nodes belonging to the fragments of the path.
More precisely, $S_{i,j}$ contains $d$ iff the node corresponding to
$r[i..|r|]r^{\alpha-j}r[1..d]$ is explicit. Then $S_{i,j}\subseteq
S_{i+1,j}$ if $i<|r|$ and $S_{i,j}\subseteq S_{1,j+1}$ if $i=|r|$,
since if $v$ is an explicit node corresponding to some
$r[i..|r|]r^{\alpha-j}r[1..d]$, then following its suffix link leads
us to an explicit node corresponding to either
$r[i+1..|r|]r^{\alpha-j}r[1..d]$ or $r[1..|r|]r^{\alpha-j-1}r[1..d]$.

Now if the answer to a query is a node at string
depth at least $\frac{3}{4}\ell$, it belongs to some $S_{i,j}$,
hence we need to preprocess all these sets for
predecessor queries. As the sets are nested, by Lemma~\ref{lemma:nested},
it requires only $\Oh(|r|\log|r|+\sum_{i,j}S_{i,j})$ words of space, which is $\Oh(|r|\log|r|+s)$,
where $s$ is at most the number of explicit nodes in $T_r$ at string depth at
least $\frac{1}{2}r$. Before we find the predecessor in the appropriate $S_{i,j}$,
we need to determine which set to query. For this we separately
store for every $i$ a pointer to the explicit node
with the largest string depth on $p_i$.
Then to determine the node corresponding to a substring $r[i..|r|]rr..$, we
first look at its length to check if the explicit node with the largest string
depth on $p_i$ should be returned. If not, with a simple division we can
determine which $S_{i,j}$ should be considered, so that the answer is either
there, or in $S_{i,j-1}$. Then locating the predecessor in these two sets
gives us the string depth of the node that we
should return. To determine its depth, we store for every nonempty $S_{i,j}$
the smallest depth of an explicit node there. These values are stored in a separate
array for every $i$. By adding the rank of the predecessor in the appropriate
set to the smallest depth of an explicit node stored there, we get the final
depth.
For every cyclic shift $r'=r[i..|r|]r[1..i-1]$ of $r$, where
$i=1,2,\ldots,|r|$, we denote by $p_i$ the longest path in $T_r$
corresponding to a prefix of $r'^{\infty}$. Hence the whole $T_r$
can be seen as a union of these $|r|$ paths. The paths are not
necessarily disjoint, but no two of them share a common prefix
of length $2|r|$, as otherwise the periodicity lemma would imply
that $r$ is actually not primitive. We conceptually
extend every $p_{i}$ so that it corresponds to $r[i..|r|]r^{\alpha}$,
with the same value of $\alpha$ for every $i$. As we are working with
a compacted trie anyway, such an extension doesn't increase the size of
the problem.

We split every $p_{i}$ into a prefix corresponding to $r[i..|r|]r^{\alpha-\beta}$
and then $\beta$ fragments corresponding to the remaining $\beta$
repetitions of $r$. The value of $\beta$ is chosen so that the following
two conditions hold:
\begin{enumerate}
\item any explicit node that we could possibly be required to return as an
answer belongs to one of these $\beta$ fragments,
\item any explicit node belonging to one of these $\beta$ fragments
is at string depth at least $\frac{1}{2}\ell$.
\end{enumerate}
The conditions translate to $|r|-i+1+|r|(\alpha-\beta)\leq\frac{3}{4}\ell$
and $|r|-i+1+|r|(\alpha-\beta)\geq\frac{1}{2}\ell$, respectively. As
$|r|\leq\frac{1}{4}\ell$, such a $\beta$ always exists.

For every $p_i$ we define the sets $S_{i,j}$ for $j=1,2,\ldots,\beta$ describing
the string depths of all explicit nodes belonging to the fragments of the path.
More precisely, $S_{i,j}$ contains $d$ iff the node corresponding to
$r[i..|r|]r^{\alpha-j}r[1..d]$ is explicit. Then $S_{i,j}\subseteq
S_{i+1,j}$ if $i<|r|$ and $S_{i,j}\subseteq S_{1,j+1}$ if $i=|r|$,
since if $v$ is an explicit node corresponding to some
$r[i..|r|]r^{\alpha-j}r[1..d]$, then following its suffix link leads
us to an explicit node corresponding to either
$r[i+1..|r|]r^{\alpha-j}r[1..d]$ or $r[1..|r|]r^{\alpha-j-1}r[1..d]$.

Now if the answer to a query is a node at string
depth at least $\frac{3}{4}\ell$, it belongs to some $S_{i,j}$,
hence we need to preprocess all these sets for
predecessor queries. As the sets are nested, by Lemma~\ref{lemma:nested},
it requires only $\Oh(|r|\log|r|+\sum_{i,j}S_{i,j})$ words of space, which is $\Oh(|r|\log|r|+s)$,
where $s$ is at most the number of explicit nodes in $T_r$ at string depth at
least $\frac{1}{2}r$. Before we find the predecessor in the appropriate $S_{i,j}$,
we need to determine which set to query. For this we separately
store for every $i$ a pointer to the explicit node
with the largest string depth on $p_i$.
Then to determine the node corresponding to a substring $r[i..|r|]rr..$, we
first look at its length to check if the explicit node with the largest string
depth on $p_i$ should be returned. If not, with a simple division we can
determine which $S_{i,j}$ should be considered, so that the answer is either
there, or in $S_{i,j-1}$. Then locating the predecessor in these two sets
gives us the string depth of the node that we
should return. To determine its depth, we store for every nonempty $S_{i,j}$
the smallest depth of an explicit node there. These values are stored in a separate
array for every $i$. By adding the rank of the predecessor in the appropriate
set to the smallest depth of an explicit node stored there, we get the final
depth.
\qed
\end{proof}

Now if $r$ and $r'$ are two different Lyndon words of length at most
$\frac{1}{4}\ell$, the sets of explicit nodes in $T_r$ and $T_{r'}$ at
string depth at least $\frac{1}{2}\ell$ are disjoint, as otherwise
from the periodicity lemma we would get that $r$ and $r'$ are cyclic
shifts of the same string. Hence if we apply the above lemma for every
different Lyndon word $r$ such that some $w_i$ has the middle part
which is a substring of $r^\infty$, all explicit nodes contributing to
the $s$ added in the space complexity will sum up to $n$. Also, all
$|r|$ will sum up to at most $\sum_i \frac{1}{4}|w_i|=\Oh(n)$, making
the total space complexity $\Oh(n\log n)$.

\section{\label{sec:space-reduction}Decreasing the space}

In this section we improve the space complexity of the solution to
$\Oh(n)$. As an intermediate step, we will first show how to make it
$\Oh(n\log n)$ by improving the solution for \PINSlong and
\PISNSlong. This almost immediately yields an improved space bound of
$\Oh(n\log n)$, as it allows us to solve long substring retrieval in
$\Oh(n)$ space by reducing the space complexity of
Lemma~\ref{lemma:periodic preprocessing} to $\Oh(|r|+s)$. Further
improvement requires more work.

We improve the solution for \PINSlong by making use of techniques from
the area of succinct data structures~\cite{J89}. We reemphasise that,
even though we do make reference to individual bits, all space bounds
are stated in words.  In our application, our goal is to spend a
constant number of words per element in our data structure, since we
desire the overall space to be linear.  The problem with the previous
solution is the extra $\Oh(N \log^2 N)$ costs in terms of the universe
size.  Here we focus on reducing the cost in terms of the universe,
$N$, by a polylogarithmic factor.  

In this section we make use of $\wordSize$ to denote the word size (in
bits) of our word-RAM.  We do this to avoid conflating the word with
the problem size, as it would become an issue in later proofs.  As in
the proof of Lemma~\ref{lemma:shrinking-nested}, we assume that we
have access to a universal table of size $\Theta(n^\varepsilon)$ to
support $\texttt{rank}$ and $\texttt{select}$ queries on small bit
vectors.

The following lemma presents a space/query time tradeoff bound for
supporting $\texttt{rank}$ and $\texttt{select}$ on bit vectors that
are weaker than that of P{\v a}tra{\c s}cu~\cite{MihaiSuccincter}.
However, our data structure is much simpler, since we do not need it
to be succinct.  Furthermore, in our problem we are also interested in
reducing the preprocessing costs, which are not discussed by P{\v
  a}tra{\c s}cu, since we eventually plan on reducing the
preprocessing time to linear.

\begin{lemma}
\label{lemma:rank-select}
A bit vector of total length $N$ bits, in which $M$ bits
are ones, can be represented by a data structure occupying $\Oh(tM +
N/\wordSize^{t})$ space, for any $t \ge 1$.  The operations $\texttt{rank}$
and $\texttt{select}$ can be performed on the bit vector in $\Oh(t)$
time. The preprocessing time is $\Oh(tM + N/\wordSize^{t})$, \emph{assuming}
the input is the list of indices of the $M$ one bits, and not
including the cost of building the universal table.
\end{lemma}

\begin{proof}
For the proof we discuss how to support $\texttt{rank}$ in constant
time using the claimed amount of space.  After accomplishing this, it
is trivial to implement $\texttt{select}$ via $\Oh(M)$ additional
space by explicitly storing the answers.

Define a \emph{packed decomposition} to be a decomposition of a
universe $[1,\wordSize^u]$, for some $u\ge q $, into buckets of size $\wordSize^{u-1}$.
Each bucket in the decomposition is assigned a $1$
bit iff it is non-empty (i.e., the bucket contains at least one set
bit).  Let $B$ be a bit vector storing the bits of the decomposition.
For each $1$ bit in $B$, e.g., $B[j] = 1$ for some $1 \le j \le c\wordSize$, we
explicitly store the number of $1$s in the range $[1,(j-1)\wordSize^{u-1}/c)$.
  That is, the partial sums up to the start of the bucket represented
  by $B[j]$.  All these numbers are stored in an array $C$ of size
  $\Oh((M' \lg \wordSize)/\wordSize) = \Oh(M')$, where $M'$ is at most the
  number of $1$ bits in the whole universe. The total space required for the
  packed decomposition, i.e., to store the arrays $B$ and $C$, is therefore $\Oh(1 + M')$.

At the separate top level of our data structure, we divide the universe $[1,N]$
into buckets of size $\wordSize^{t}$.  For each bucket we explicitly store
partial sums, counting the number of ones up to the start of that
bucket in an array.  Overall, this takes space $\Oh(N/ \wordSize^{t})$.  For
each non-empty bucket of length $\wordSize^{t}$, we store a packed
decomposition on that bucket.  We recursively store packed decompositions
of the non-empty buckets until the universe is of size $\wordSize$.
Visualising the decomposition as a tree, we see the leaves are of size
$\wordSize$, the height of the tree is $t$, and every leaf contains
at least one $1$, hence the total number of internal nodes in all trees
is $tM$. At each such internal node, we store an array of pointers, each
one corresponding to a non-empty bucket and pointing to a packed
decomposition in the lower level.  Overall, this adds an additional
$\Oh(tM)$ space cost.

The space of the data structure is no more than $\Oh(tM +
N/\wordSize^{t})$ based on the arguments above. We can perform
\texttt{rank} queries by recursing down the tree at most $\Oh(t)$
levels, and computing the number of ones up to the range represented
by the tree node.  At each level in the tree this takes $\Oh(1)$ time,
since: at the top level the values are explicitly stored; at each
internal node we can in $\Oh(1)$ time determine which entry in $C$ to
add to the running total by counting the one bits in $B$ using table
lookup; and, finally, we can also count the number of one bits up to
the search position in the leaf using table lookup.

Next, we discuss the preprocessing costs.  Assume we get the input as
a list of indices of the one positions.  We recursively bucket sort
the indices to construct the tree (i.e., the pointer structure),
placing each index in its appropriate leaf in time $\Oh(tM
+N/\wordSize^t)$.  After we have the pointer structure, it is trivial
to construct the remaining data structures---partial sums and bit
vectors---via a preorder traveral of the tree, using no more than the
claimed time bound.
\qed
\end{proof}

Using the previous lemma, we can trivially answer predecessor queries
on a set $S_i$, by representing it as a bit vector and using a
\texttt{rank} query followed by a \texttt{select} query.

\begin{lemma}\label{lemma:space-efficient-nested}
\PINS can be solved using a data structure that occupies
$\Oh(N/\wordSize^{t_1} + \sum_{i}|S_{i}|)$ space and performs queries
in $\Oh(1)$ time, for any constant $t_1 \ge 1$.  The preprocessing
time is $\Oh(N/\wordSize^{t_1} + \sum_{i}|S_{i}|)$.
\end{lemma}

\begin{proof}
We follow the same general strategy as Lemma~\ref{lemma:nested}. We
note that to prove the lemma it suffices to reduce the space of the
predecessor data structures for the sets $S_{g_1}, ..., S_{g_{\log
    N}}$ to $\Oh(N/\wordSize^{t_1} + \sum_{i}|S_{i}|)$ space: the
remaining predecessor data structures for the other sets occupy no
more than $\Oh(\sum_{i}|S_i|)$ space by replacing them by
Lemma~\ref{lemma:rank-select}.  We store set $S_{g_{\log N}}$ in the
data structure of Lemma~\ref{lemma:rank-select}, constructed over the
universe $[1,N]$, with parameter $t_0$ (which will be fixed later).
Each set $S_{g_j}$, $j\in[1,\log N-1]$ is represented using the data
structure of Lemma~\ref{lemma:rank-select} for the universe
$[1,|S_{g_{\log N}}|]$, marking the elements in $S_{g_{\log N}}$ which
are present in $S_{g_j}$ with $1$ bits, again for the parameter $t_0$.
Together these structures allow us to locate the predecessor of an
element in $S_{g_j}$, and occupy at most $\Oh(N \log N /
\wordSize^{t_0} + \sum_i|S_i|)$ space.  Thus, by ensuring $t_0 =
t_1+1$, we get the claimed space bound.

For the preprocessing time, we observe that computing the groups can
be done in time $\Oh(\sum_i|S_i|)$.  After computing the groups, we
can construct the predecessor data structure of
Lemma~\ref{lemma:rank-select} in time proportional to their space.
\qed
\end{proof}

Using the above result, we can improve the solution for \PISNS.

\begin{lemma}
\PISNS can be solved using $\Oh(N/\wordSize^{t_2} +\sum_{i}|S_{i}|)$
space for the preprocessing and performs queries in $\Oh(1)$ time, for
any constant $t_2 \ge 1$.  The preprocessing time is
$\Oh(N/\wordSize^{t_2} + \sum_i |S_i|)$.
\end{lemma}

\begin{proof}
Let $k$ denote the total number of sets.  Combining
Lemmas~\ref{lemma:shrinking-nested}
and~\ref{lemma:space-efficient-nested} and setting the parameter $t_1
= t_2 +1$ yields the desired space bound.

For the preprocessing time, we first show that the decomposition of
the input sets into their respective levels can be computed in the
claimed time bound. This can be done by using radix sort on all the
input sets simultaneously.  This takes $\Oh(N^\varepsilon + \sum_i
|S_i|)$ time, for any constant $0 < \varepsilon < 1$.  We also compute
the maximum element of each set $S_i$, denoted $n_i$, which can be
done within the same time bound.  Using this information together with
the sorted lists, we construct the decomposition recursively in the
following way.  We assume $N$ is a power of two to make the analysis
simpler, and start the algorithm at level $\ell = 1$, $N_L = 1$ and
$N_H = N$.

\begin{enumerate}
\item If $N_L = N_H$ or $k = 1$, construct the \PINS
  subproblem and exit.
\item Otherwise, binary search for the set $n_{k'}$, such that $n_{k'}
  \ge (N_L+N_H)/2$ and $(N_L+N_H)/2 > n_{k' +1}$.
\item \label{enum-step:remove} Scan the sorted lists $S_{k'},
  S_{k'-1}, ...$, removing the elements that are in the range
  $[N_L,(N_L+N_H)/2]$.  Once we encounter a list
  $S_{j}$ containing no elements in
  $[N_L,(N_L+N_H)/2]$ we stop.
\item \label{enum-step:guide-bv} Set bit $\ell$ in the guide vectors
  for $S_{j},...,S_{k'}$.  
\item \label{enum-step:PIMS} Construct the \PINS instance on the
  removed elements (if there are any) with a universe
  $[N_L,N_{H}]$.  Also store a pointer to this
  subproblem for each guide bit that was set in the previous step.
\item Recurse on the sets $S_1, ..., S_{k'}$ with level $\ell + 1$,
  $N_L = (N_L+N_H)/2 + 1$, and $N_H = N_H$.
\item Recurse on the sets $S_{k'+1}, ..., S_{k}$ with level $\ell +
  1$, $N_L = N_L$ and $N_H = (N_L+N_H)/2$.
\end{enumerate}

The cost of steps \ref{enum-step:remove} and \ref{enum-step:guide-bv}
is bounded by the number of elements passed to the \PINS subproblem,
due to the nesting property of the sets.  By
Lemma~\ref{lemma:space-efficient-nested} the cost of steps
\ref{enum-step:remove}, \ref{enum-step:guide-bv}, and
\ref{enum-step:PIMS} over the \emph{entire} algorithm is clearly no
more than $\Oh(N/\wordSize^{t_2} + \sum_i |S_i|)$, since each range of
the universe appears in at most $\log N$ levels, and
$\wordSize\ge\log N$.  Thus, we need only analyse the cost of the
remaining steps.  Bounding the cost of the binary search in terms of
the number of sets $k$, we get the recurrences:
\begin{eqnarray*}
T(1,k) &=& \Theta(1)\\
T(N,1) &=& \Theta(1)\\
T(N,k) &\le& \Theta(\log \min(k',k-k')) + T(N/2, k') + T(N/2, k - k')
\end{eqnarray*}
\noindent
Notice that want to bound the time of a single search by $\Oh(\log\min(k',k-k'))$,
which requires starting it simultaneously from both ends.
Then, one can choose coefficients $\alpha$ and $\beta$ so that
$T(N,k) \le \alpha k - \beta \sqrt{k}$. This is because, by induction,
we only need to bound $$\log \min(k',k-k') + \alpha k' -\beta \sqrt{k'} + \alpha(k-k') - \beta \sqrt{k-k'}$$
which is maximised for $k'=k/2$, and we can always choose 
$\beta$ large enough so that $\log k -2\beta \sqrt{k/2} \le -\beta \sqrt{k}$.
Then we select $\alpha$ large enough so that the base of the induction
holds. Thus, we get that $T(N,k) \le \Theta(k) \le \Oh(\sum_i|S_i|)$, and the
overall cost is therefore $\Oh(N/\wordSize^{t_2} + \sum_i(|S_i|))$.
\qed
\end{proof}

This gives us the basic tools needed to improve the space complexity
of the whole algorithm.  Now we need to carefully look at all of its
components.  First of all, we need to decrease the space bound in
Lemma~\ref{lemma:periodic preprocessing}.  We would like to reduce it
to $\Oh(|r|/\wordSize+s)$.  By plugging in the better implementation
of \PINSlong, we can reduce the bound to $\Oh(|r|/\wordSize^{t}+s)$,
plus the space needed to store, for every $i$, the pointer to the
explicit node with the largest string depth on $p_{i}$. By storing the
pointer only if the corresponding explicit node is at string depth at least
$\frac{1}{2}\ell$ and recalling that no two paths can share
a prefix of length $\frac{1}{2}\ell$, we decrease the total space used
by the pointers to $\Oh(|r|/\wordSize+s)$.

The second step is to relax the definition of long substring retrieval.
Recall that the goal was to preprocess a generalised suffix tree
built for a collection of $\beta=\Oh(n/\ell)$ documents $w_{1},w_{2},\ldots,w_{\beta}$,
all of the same length $\ell$, so that we can retrieve the node corresponding
to any $w_{i}[j..k]$ of length at least $\frac{3}{4}\ell$.
In {\it generalised long substring retrieval}, we consider a collection
of $\beta=\Oh(n/\ell)$ documents $w_{1},w_{2},\ldots,w_{\beta}$, all of length
at most $\ell$. We want to preprocess the bottom part of their
generalised suffix tree consisting of all nodes at string depth
at least $\frac{3}{4}\ell$, so that given a pointer to a leaf at
string depth at least $\frac{3}{4}\ell$, we can perform a predecessor search among
all of its explicit ancestors in the bottom part. Hence the difference
between the non-generalised and generalised version is that we allow
some of the strings to be shorter, and we assume that we are given
a pointer to a leaf as opposed to just the numbers $i,j,k$.

\begin{lemma}
\label{lemma:better long substring retrieval}
After $\Oh(n/\wordSize+n/\ell+s)$ space preprocessing,
where $s$ is the number of explicit nodes at string depth at least
$\frac{1}{2}\ell$ in the generalised suffix tree,
generalised long substring retrieval can be solved in $\Oh(1)$ time.
\end{lemma}

\begin{proof}
We separately preprocess all active and non-active nodes of the
generalised suffix tree. For the active nodes, by plugging in the
improved solution for \PISNSlong, we decrease the space usage to
$\Oh(n/\wordSize+s_{a})$, where $s_{a}$ is the number of active
explicit nodes. For the remaining nodes, for every Lyndon word $r$
such that at least one $w_{i}[\frac{1}{4}\ell..\frac{3}{4}\ell]$ is a
substring of $r^{\infty}$, we need $\Oh(|r|/\wordSize+s_{r})$ space,
where $s_{r}$ is the number of explicit nodes at string depth at least
$\frac{1}{2}\ell$ in $T_{r}$, where $T_{r}$ is the subtree of the
whole generalised suffix tree corresponding to all substrings of
$r^{\infty}$ of length at least $\frac{1}{2}\ell$. As mentioned
before, all these $s_{r}$ sum up to at most the number of explicit
nodes at string depth at least $\frac{1}{2}\ell$ in the generalised
suffix tree, hence the total space complexity for all such $r$ is
$\Oh(n/\wordSize+s)$, where $s$ is the number of explicit nodes at
string depth at least $\frac{1}{2}\ell$ in the generalised suffix
tree. Additionally, we need to store for every leaf of the generalised
suffix tree its active ancestor with the smallest string depth, and
for every $w_{i}$ its corresponding $r$, if any (more precisely, a
pointer to the structure corresponding to this $r$, and furthermore
the position of some occurrence of $r$ in
$w_{i}[\frac{1}{4}\ell..\frac{3}{4}\ell]$). The former requires
$\Oh(s)$ and the latter $\Oh(n/\ell)$ space, respectively.
\qed
\end{proof}

Now we modify Lemma~\ref{lemma:reduction}. Recall that the idea
there was that, for every $\ell=\alpha 2^{k}$, where $\alpha=8,9,\ldots,15$,
we create an instance of long substring retrieval with $\Oh(n/\ell)$
documents of length $\ell$. Now we would like to say that
preprocessing each of these instances with Lemma~\ref{lemma:better long substring retrieval}
ensures that the total cost is just $\Oh(n)$, because all values of $s$
sum up to $\Oh(n)$. Unfortunately, this is not true, as an explicit node
in the generalised suffix tree built for all $w_{i}$ is not necessarily
an explicit node in the suffix tree of the whole $w$. Indeed, all
leaves of the generalised suffix tree are explicit there, but don't appear
in the suffix tree.

To fix this issue, we appropriately shorten every $w_{i}$. We choose its longest
suffix $w_{i}[j..\ell]$ such that its corresponding node in the suffix tree
has at least one explicit ancestor at string depth at least $\frac{3}{4}\ell$.
If there is no such $j\geq\frac{3}{4}\ell$, we remove $w_{i}$ from our collection, and otherwise
replace it with $w_{i}[j..\ell]$. Then for any even shorter suffix $w_{i}[j'..\ell]$, such
that $j' \geq \frac{1}{2}\ell$, the corresponding node in the suffix tree
has at least one explicit ancestor at string depth at least $\frac{1}{4}\ell$.
Hence the total number of leaves at string depth at least $\frac{1}{2}\ell$ in
the generalised suffix tree built for all shortened $w_{i}$'s can be upper
bounded by the total number of explicit nodes at string depths between
$\frac{1}{4}\ell$ and $\ell$ in the suffix tree. Bounding the number of
leaves also gives us a bound on the total number of explicit nodes in
the bottom part of the generalised suffix tree, therefore now for every $\ell$
we can bound the required space by $\Oh(n/\wordSize+n/\ell+s)$, where
$s$ is the number of explicit nodes at string depth between $\frac{1}{2}\ell$
and $\ell$ in the suffix tree. Because the values of $\ell$ are exponentially
decreasing, the sum of all these values of $s$ is then at most $\Oh(n)$,
resulting in the final bound of $\Oh(n)$ on the required space.

Finally, we describe how to answer a query using the structures for generalised
substring retrieval built for the shortened strings. For this, given a substring $s$ of $w$,
we need to access an appropriately chosen instance of generalised substring
retrieval, and also locate the leaf of the corresponding generalised suffix tree.
Recall that in the proof of Lemma~\ref{lemma:reduction} we were able to find
the instance by simply computing $\alpha\in\{8,9,\ldots,15\}$ and $k$ such that
$(\alpha-2)2^{k}\leq |s|<(\alpha-1)2^{k}$. Now, however, it might be the case
that the only $w_{i}$ in the instance containing $s$ as a substring has
been shortened, hence we cannot use it to retrieve the node corresponding to
$s$ in the suffix tree. In such case, though, it must be an implicit node
lying on a relatively long edge, i.e., an edge from a node at string depth
at least $\ell$ to a node at string depth at most $\frac{3}{4}\ell$. This
suggest a simple fix: for every $\ell=\alpha 2^{k}$, where $\alpha\in\{8,9,\ldots,15\}$,
we mark all the explicit nodes of the suffix
tree, such that their string depth is between $\frac{3}{4}\ell$ and $\ell$, but
all their (explicit) descendants have string depth exceeding $\ell$.  Then,
on any path from a leaf to the root, at most a single explicit node for every such $\ell$
is marked, hence using Lemma~\ref{lemma:marked predecessor} we can
preprocess all these marked nodes in $\Oh(n)$ space, so that given
a leaf in the suffix tree we can search for the predecessor among its
marked ancestors in $\Oh(1)$ time. Notice that the same explicit node $v$ might
be marked because of multiple values of $\ell$. Nevertheless, there is
a constant number of such relevant values of $\ell$.
For every such value, we store a pointer to the corresponding instance
of generalised long substring retrieval at $v$, and also a pointer to
any leaf in the subtree of the node corresponding to $v$ in the bottom
part of the generalised suffix tree constructed for the instance.

To locate the node corresponding to $s=w[i..j]$ in the suffix tree,
we first execute a predecessor search among the marked ancestor of the
leaf corresponding to $w[i..n]$. As a result, we get a marked node $u$
belonging to the subtree of $v$, such there are no marked nodes between
$v$ and $u$. Now the first possibility is that $v$ lies on the edge
from $u$ to its parent. If not, then $v$ lies on an edge from some $v'$
to its parent (possibly, $v'=v$), where $v'$ is an ancestor of $u$. Furthermore,
all nodes between $u$ and $v'$, including $v'$, are not marked. 
But then the string depth of $v'$ must be quite similar to the string depth
of $u$. More precisely, if the string depth of $u$ is between $\frac{3}{4}\ell$
and $\ell$ for some $\ell=\alpha 2^{k}$, where $\alpha\in\{8,9,\ldots,15\}$, then
the string depth of $u'$ must be within the same range. Otherwise, i.e.,
if the string depth of $u'$ was smaller than $\frac{3}{4}\ell$, then we could
find $\ell' = \alpha' 2^{k'}$, where $\alpha'\in\{8,9,\ldots,15\}$ and $k'<k$, such
that the string depth of $u'$ is between $\frac{3}{4}\ell'$ and $\ell'$,
hence some node between $u'$ would have been marked, which is absurd.
Therefore, using the pointers stored at $u$, we reduce the question to
generalised long substring retrieval, which can be solved in $\Oh(1)$ time.

\bibliographystyle{splncs03}
\bibliography{biblio}

\newpage
\appendix

\section{\label{sec:lb}Lower bound for arbitrary trees with small weights}

In this section we show that answering weighted ancestor queries on an
\emph{arbitrary tree} requires $\Omega(\lg \lg n)$ time using a data
structure of size $\Oh(n \polylogws( n))$, \emph{even if the node
  weights are bounded by $n$}.

Assume that, given such a tree on $n$ nodes, we can construct a data
structure occupying $\Oh(n \polylogws(n))$ space weighted ancestor
structure supporting queries in $t(n)$ time. Then we can construct a
predecessor structure for $n$ elements, drawn from the universe
$[1,n^2]$, that occupies $\Oh(n \polylogws(n))$ space and supports
queries in $t(n)$ time as follows.

\begin{enumerate}
\item \label{enum-step:pred-block} Split the universe into equally
  sized blocks of length $n$, and for each block $[x(n-1) + 1,xn]$,
  where $1 \le x \le n$, explicitly store the element that is the
  predecessor $x(n-1)$.  This adds $\Oh(n)$ space overall.

\item \label{enum-step:pointer-block} For each block create a separate
  path with string depth $n$ and put the elements contained in the
  block on the path. For each non-empty block, we also store a pointer
  to the lowest node in this path in an array.  As before this array
  will take an additional $\Oh(n)$ space.

\item To answer a predecessor query, first find the block containing
  the query element.  If a pointer is stored for this block in
  step~\ref{enum-step:pointer-block}, we do a weighted ancestor query
  on the corresponding path. If no pointer is stored, or the result of
  the weighted ancestor query is the root---which we consider a dummy
  node, and interpret as meaning that no predecessor exists in the
  path---then we return the result stored in the array from the
  step~\ref{enum-step:pred-block}.
\end{enumerate}

Immediately, since the universe is at least
$\Omega(n^{1+\varepsilon})$ and the space occupied by the data
structure is $\Oh(n\polylogws(n))$, the query must take $\Omega(\log\log
n)$ time by the lower bound of P\u{a}tra\c{s}cu and
Thorup~\cite{PatrascuPredecessor}.

\end{document}